\documentclass[conference,twoside]{IEEEtran}
\usepackage{graphicx}
\usepackage{epstopdf}
\DeclareGraphicsExtensions{.eps,.pdf} 
\graphicspath{ {figures/} }
\usepackage{cite}
\usepackage{subfigure}
\usepackage{amssymb,amsmath}
\usepackage{algorithm}
\usepackage{algorithmic}
\usepackage{color}
\usepackage{multirow}
\usepackage{multicol}
\usepackage{cases}
\usepackage{setspace}

\makeatletter
\newcommand\restartchapters{\par
  \setcounter{chapter}{0}%
  \setcounter{section}{0}%
  \gdef\@chapapp{\chaptername}%
  \gdef\thechapter{\@arabic\c@chapter}}
\makeatother

\newtheorem{theorem}{\bf {Theorem}}

\newtheorem{definition}{\bf {Definition}}

\newcommand{\tr}{{\mathtt{Tr}}}
\newcommand{\diag}{{\mathtt{diag}}}

\newcommand{\st}{{\mathrm{s.t.}}}

\newcommand{\AP}{\mathtt{AP}_{m}}

\newcommand{\RSI}{\mathtt{RSI}}
\newcommand{\APx}[1]{\mathtt{AP}_{#1}}
\newcommand{\ULU}{\mathtt{U}_{\ell}^{\mathtt{u}}}
\newcommand{\ul}{\mathtt{u}}
\newcommand{\SI}{\mathtt{SI}}
\newcommand{\AtoA}{\mathtt{AA}}
\newcommand{\dl}{\mathtt{d}}
\newcommand{\DLU}{\mathtt{U}_{k}^{\mathtt{d}}}
\newcommand{\DLUi}[1]{\mathtt{U}_{#1}^{\mathtt{d}}}

\newcommand{\cM}{\mathcal{M}}
\newcommand{\cK}{\mathcal{K}}
\newcommand{\cL}{\mathcal{L}}
\newcommand{\tZF}{\mathtt{ZF}}

\newcommand*{\hili}{\color{black}}

\newcommand*{\hilidraf}{\color{black}}

\usepackage{capt-of}
\usepackage{dblfloatfix}
\usepackage{varwidth}
\makeatletter
\g@addto@macro\normalsize{%
 \setlength\abovedisplayskip{1pt}
 \setlength\belowdisplayskip{1pt}
 \setlength\abovedisplayshortskip{1pt}
 \setlength\belowdisplayshortskip{1pt}
}

\newcommand{\subparagraph}{}

\usepackage{titlesec}

\makeatletter
\def\set@curr@file#1{%
	\begingroup
	\escapechar\m@ne
	\xdef\@curr@file{\expandafter\string\csname #1\endcsname}%
	\endgroup
}
\def\quote@name#1{"\quote@@name#1\@gobble""}
\def\quote@@name#1"{#1\quote@@name}
\def\unquote@name#1{\quote@@name#1\@gobble"}
\makeatother

\begin{document}
\bstctlcite{IEEEexample:BSTcontrol}

\title{{\hili \huge A Novel Heap-based Pilot Assignment for Full Duplex Cell-Free Massive MIMO with Zero-Forcing}\vspace{-3pt}} 
\author{
	\IEEEauthorblockN{Hieu V. Nguyen$^{1}$, Van-Dinh Nguyen$^{3}$, Octavia~A.~Dobre$^{2}$, Shree Krishna Sharma$^{3}$, \\ Symeon Chatzinotas$^{3}$, Bj$\ddot{\text{o}}$rn Ottersten$^{3}$,  and Oh-Soon Shin$^{1}$ }\vspace{2pt}
	\IEEEauthorblockA{\normalsize $^{1}$School of Electronic Engineering \& Department of ICMC Convergence Technology, Soongsil University, Korea \\
		$^{2}$Faculty of Engineering and Applied Science, Memorial University, St. John's, NL, Canada\\
		$^{3}$SnT, University of Luxembourg, L-1855 Luxembourg City,
		Luxembourg \\ 
		\vspace{-40pt}}
	\thanks{This work was supported in part by the ERC project AGNOSTIC and  FNR ECLECTIC.}
}
\maketitle
\begin{abstract}
This paper investigates the combined benefits of  full-duplex (FD) and cell-free massive multiple-input multiple-output (CF-mMIMO), where a large number of distributed  access points (APs) having FD capability simultaneously serve numerous uplink and downlink user equipments (UEs) on the same time-frequency resources.
To enable the incorporation of FD technology in CF-mMIMO systems, we  propose a novel heap-based pilot assignment algorithm, which not only can mitigate the effects of pilot contamination but also reduce the involved computational complexity. Then, we formulate a robust design problem for spectral efficiency (SE) maximization in which the power control and AP-UE association are jointly optimized, resulting in a difficult mixed-integer nonconvex programming. To solve this problem, we derive a  more tractable problem before developing a very simple iterative algorithm based on inner approximation method with polynomial computational complexity. Numerical results show that our proposed methods with realistic parameters significantly outperform the existing approaches in terms of the quality of channel estimate and SE.
\end{abstract}


\section{Introduction} \label{Introduction} 

In-band full-duplex (FD) has been envisaged as a key enabling technology to increase the spectral efficiency (SE) of a wireless link over its half-duplex (HD) counterparts by a factor close to two, since it allows downlink (DL) and uplink (UL) transmissions on the same time-frequency resources \cite{Yadav:IEEEVTM:June2018,Sharma:IEEECST:2018, Sabharwal:JSAC:Feb2014,GoyalCMag15}. Although the main barrier in implementing FD  is the self-interference (SI), many recent advances in  active and passive SI suppression techniques have been successful to  bring the SI power at the background noise level \cite{Bharadia14}. As a result,  FD-enabled base station (BS) systems  have been widely studied in small-cell (SC) cellular networks \cite{Yadav:Access,Dinh:Access,Hieu:IEEETWC:June2019,Hieu:IEEETCOM:June2019, Dinh:JSAC:18}.

Recently,  cell-free massive multiple-input multiple-output (CF-mMIMO), where a very large number of access points (APs) is distributed over a wide  area to coherently serve numerous user equipments (UEs) in the same resources, has been proposed to overcome the inter-cell interference \cite{Ngo:TWC:Mar2017,Bashar:IEEETWC:Apr2019}. CF-mMIMO not only inherits the properties of favorable propagation and channel hardening from collocated massive MIMO networks (Co-mMIMO), but also reduces path losses due to the APs placed closer to UEs. It can be seen that the performance gains of CF-mMIMO are obtained by the joint process at a central processing unit (CPU). 

Despite their potentials, there are only a few attempts on characterizing the performance of an FD-enabled CF-mMIMO. In this regard, the authors in \cite{Vu:ICC:May2019} analyzed the performance of FD CF-mMIMO, where all APs operate in the FD mode with the usage of a conjugate beamforming/matched filtering transmission design. Tackling the imperfect channel state information (CSI) and spatial correlation for the FD CF-mMIMO system was studied in \cite{Wang:arxiv:2019}, with a genetic algorithm-based user scheduling strategy applied to alleviate the co-channel interference (CCI). However, none of the aforementioned works properly examined the training procedure for a practical purpose in FD CF-mMIMO, while the existing approaches designed for HD operation may not be inapplicable to FD operation under the strong effects of SI, inter-AP interference (IAI) and CCI. Moreover, it is crucial to develop a low-complexity robust design to attain the optimal SE performance of FD CF-mMIMO systems under channel uncertainties.

In the above context, this paper considers an FD CF-mMIMO system under time-division duplex (TDD) operation, where  FD-enabled multiple-antenna APs  simultaneously serve UL and DL UEs on the same time-frequency resources. A novel heap-based pilot assignment strategy for FD CF-mMIMO is proposed not only to enhance the quality of channel estimates but also to reduce the computational complexity required for the training process. {\hilidraf Furthermore, we develop an efficient transmission design for FD CF-mMIMO, in which the power control and AP-DL UE association are jointly optimized to reduce network interference.} Based on the widely-used zero-forcing (ZF) method, we propose a new robust design and an iterative algorithm whose complexity is less impacted by the large number of APs (or antennas). The main contributions of this paper are summarized as follows:
\begin{itemize}
   \item We first propose a novel heap-based pilot assignment algorithm to reduce both pilot contamination and training complexity.
	\item Aiming at SE optimization, we introduce new binary variables to establish the AP-DL UE associations. In our system design, the APs can be  automatically switched between the FD and HD operations, which allows to exploit the full potential of FD CF-mMIMO. 
	\item We further consider a robust design problem for SE maximization with joint power control and AP-DL UE association, which is formulated as a difficult class of mixed-integer nonconvex optimization problem. To efficiently solve the problem, we first exploit the structure of the optimal solution of binary variables to transform the original problem into a more tractable form. Based on the ZF design, a very simple, yet efficient algorithm based on the inner approximation (IA) method is devised to obtain a locally optimal solution of the problem.
	\item Numerical results confirm  that the proposed algorithms greatly reduce the pilot contamination and  improve the SE performance over the current state-of-the-art approaches under both HD and FD operation modes.
\end{itemize}


\emph{Notation}: 
$\tr(\cdot)$, $\|\cdot\|$ and $|\cdot|$ are the trace,  Euclidean norm and  absolute value, respectively.  $\mathtt{diag}(\mathbf{a})$ returns  the diagonal matrix with the main diagonal assembled from elements of $\mathbf{a}$.   

\section{System Model} \label{SystemModelandProblemFormulation}
\subsection{Transmission Model}
\begin{figure}[t]
	\centering
	\includegraphics[width=0.76\columnwidth,trim={0cm 0.0cm 0cm 0.0cm}]{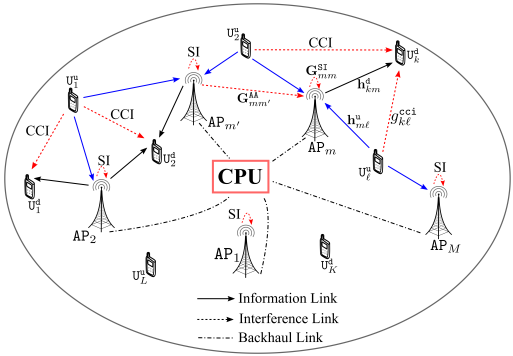}
		\vspace{-5pt}
	\caption{An illustration of the full-duplex cell-free massive MIMO system.}
	\label{fig: system model}
\end{figure}
An FD CF-mMIMO system operated in TDD mode is considered, where the set $\mathcal{M}\triangleq\{1,2,\cdots,M\}$ of $M$ FD-enabled APs simultaneously serves the sets  $\mathcal{K}\triangleq\{1,2,\cdots,K\}$ of $K$  DL UEs and $\mathcal{L}\triangleq\{1,2,\cdots,L\}$ of $L$  UL UEs in the same time-frequency resources, as illustrated in Fig. \ref{fig: system model}. The total number of  APs' antennas is $N=\sum_{m\in\mathcal{M}}N_m $, where $N_m$ is the number of antennas at AP $m$, while each UE has a single-antenna. All APs are equipped with  FD capability by circulator-based FD radio prototypes \cite{Bharadia14}, which are connected to the CPU through   perfect backhaul links with sufficiently large capacities \cite{Ngo:TWC:Mar2017}. At the CPU, the message sent by an UL UE is decoded by aggregating  the received signals from all active APs due to the UL broadcast transmission. Meanwhile, each DL UE should be served by a subset of active APs with good channel conditions \cite{Ngo:TWC:Mar2017}. This is done by introducing new binary variables to establish the AP-DL UE associations.


For notational convenience, let us denote the $m$-th AP,  $ k $-th DL UE and $\ell $-th UL UE  by $\AP$, $\DLU$ and $\ULU$, respectively. The channel vectors and matrices from  $ \AP  \rightarrow\DLU$, $\ULU\rightarrow \AP $, $\ULU\rightarrow\DLU$ and $ \APx{m'}\rightarrow\APx{m},\forall m'\in\cM $ are denoted by $ \mathbf{h}_{km}^\dl\in \mathbb{C}^{1\times N_m} $, $ \mathbf{h}_{m\ell}^\ul \in \mathbb{C}^{N_m\times1}$, $ g_{k\ell}^{\mathtt{cci}}\in \mathbb{C} $ and $ \mathbf{G}_{mm'}^{\AtoA} \in \mathbb{C}^{N_m\times N_{m'}} $, respectively. Note that $ \mathbf{G}_{mm}^{\AtoA}$ is the SI channel at $\AP$, while $ \mathbf{G}_{mm'}^{\AtoA}, \forall m\neq m'$ is referred to as the inter-AP interference (IAI) channel since   UL signals received at $ \APx{m} $ are corrupted by  DL signals sent from $ \APx{m'} $. To differentiate the residual SI and IAI channels, we model $ \mathbf{G}_{mm'}^{\AtoA}$ as follows:
\begin{numcases}{ \mathbf{G}_{mm'}^{\AtoA}=}
\sqrt{\rho_{mm}^{\mathtt{RSI}}}\mathbf{G}_{mm}^{\mathtt{SI}}, & if $m=m'$, \nonumber\\
\mathbf{G}_{mm'}^{\AtoA}, & otherwise, \nonumber
\end{numcases}
where $\mathbf{G}_{mm}^{\mathtt{SI}}$ denotes the fading loop channel at $\AP$ which interferes the UL reception due to the concurrent DL transmission, and $\rho_{mm}^{\mathtt{RSI}}\in[0,1)$ is the residual SI suppression (SiS) level after all real-time cancellations in analog-digital domains \cite{Sabharwal:JSAC:Feb2014,Dinh:JSAC:18,Hieu:IEEETWC:June2019}. The fading loop channel $\mathbf{G}_{mm}^{\mathtt{SI}}$ can be  characterized as the Rician probability distribution
with a small Rician factor \cite{Duarte:TWC:12}, while other channels are generally modeled as $ \mathbf{h}=\sqrt{\beta}\mathbf{\ddot{h}} $ with $ \mathbf{h}\in\{\mathbf{G}_{mm'}^{\AtoA}, \mathbf{h}_{km}^{\dl}, \mathbf{h}_{m\ell}^{\ul}, g_{k\ell}^{\mathtt{cci}} \} $  accounting for the effects of large-scale fading $\beta$ (i.e., path loss and shadowing) and small-scale fading $\mathbf{\ddot{h}} $ whose elements are independent and identically distributed (i.i.d.) $\mathcal{CN}(0,1)$ random variables (RVs). Since the transmit and receive antennas as well as the APs are generally fixed in a given area without mobility, the IAI and SI channels $ \mathbf{G}_{mm'}^{\AtoA},\forall m,\;m'\in\mathcal{M} $, are assumed to be perfectly acquired at the CPU. Therefore, we focus on the channel estimation of DL, UL and CCI in the rest of paper.


\subsection{Pilot Assignment Problem for Uplink Training}
We assume that all UEs share the same orthogonal set of pilots, and the DL and UL UEs send  training sequences in different intervals to allow the channel estimation of CCI links. Let $ \tau<\min\{K,L\} $ be the length of pilot sequences. Then, the pilot set is defined as $ \boldsymbol{\Xi}\triangleq[ \boldsymbol{\varphi}_{1},\;\cdots\;,\boldsymbol{\varphi}_{\tau} ]\in\mathbb{C}^{\tau \times\tau} $, where $ \boldsymbol{\varphi}_i\in\mathbb{C}^{\tau\times 1}$ satisfies the orthogonality, i.e., $\boldsymbol{\varphi}_{i}^H\boldsymbol{\varphi}_{i'}=1$ if $i= i' \in \mathcal{T}_{\mathtt{p}}\triangleq\{1,\cdots,\tau \} $, and $\boldsymbol{\varphi}_{i}^H\boldsymbol{\varphi}_{i'}=0$, otherwise.
We introduce the assignment variable $ \upsilon_{ij}\in\{0,1\} $ to determine whether the $ i $-th pilot sequence is assigned to the $ j $-th UE, with $ j\in\mathcal{T}_{\ul}\triangleq\{1,\cdots,U\} $ and  $ U\in\{K,L\} $. As a result, the pilot assigned to  UE $j$ can be expressed as	$ \boldsymbol{\bar{\varphi}}_j = \boldsymbol{\varphi}_i $ if $ \upsilon_{ij}=1 $. Let $ \boldsymbol{\bar{\Xi}}\triangleq[ \boldsymbol{\bar{\varphi}}_{1},\cdots,\boldsymbol{\bar{\varphi}}_{U} ]\in\mathbb{C}^{\tau \times U} $ be the pilot assignment matrix, such as $\boldsymbol{\bar{\Xi}} = \boldsymbol{\Xi}\boldsymbol{\Upsilon},$
where $ \boldsymbol{\Upsilon}\triangleq[\upsilon_{ij}]_{i\in\mathcal{T}_{\mathtt{p}},j\in\mathcal{T}_{\ul} }\in\mathbb{C}^{\tau\times U} $ following by the condition:
$\sum_{i\in\mathcal{T}_{\mathtt{p}} } \upsilon_{ij} \leq 1,\;\forall j\in\mathcal{T}_{\ul}.$

The training procedure for FD CF-mMIMO in TDD operation is executed in two phases. In the first phase, UL UEs send their pilot signals to APs to perform channel estimation, and at the same time, DL UEs also receive UL pilots to estimate CCI channels. In the second phase, DL UEs send their pilot signals along with the estimates of CCI links to APs. The training signals received at $ \AP $ can be written as $\mathbf{Y}_{m}^{\mathtt{tr}} = \sum_{j'\in\mathcal{T}_{\ul}} \sqrt{\tau p_{j'}^{\mathtt{tr}}}\boldsymbol{\bar{\varphi}}_{j'} \mathbf{h}_{mj'} + \mathbf{Z}_m,$
where $ \mathbf{h}_{mj}\in\{\mathbf{h}_{km}^{\dl}, (\mathbf{h}_{m\ell}^{\ul})^H \}\in\mathbb{C}^{1\times N_m} $, and $ p_{j}^{\mathtt{tr}} $ and $ \mathbf{Z}_{m}\sim \mathcal{CN}(0,\sigma^2_{\mathtt{AP}}\mathbf{I}) $ denote the UL training power of UE $j$ and the AWGN, respectively.  Using the linear minimum mean square error (LMMSE) estimation, the channel estimate of $ \mathbf{h}_{mj} $ is given as
\begin{align} \label{eq: channel estimate}
\mathbf{\hat{h}}_{mj} & = \frac{\sqrt{\tau p_{j}^{\mathtt{tr}} }\beta_{mj}}{\sum_{j'\in\mathcal{T}_{\ul}} \tau p_{j'}^{\mathtt{tr}} \beta_{mj'} |\boldsymbol{\bar{\varphi}}_{j}^H\boldsymbol{\bar{\varphi}}_{j'}|^2 + \sigma^2_{\mathtt{AP}}} \boldsymbol{\bar{\varphi}}_{j}^H\mathbf{Y}_{m}^{\mathtt{tr}},
\end{align}
where $\beta_{mj}\in\{\beta_{km}^{\dl},\beta_{m\ell}^{\ul}\} $ is  the large-scale fading of the link between $ \AP $ and UE $j$.
We denote $\mathbf{\tilde{h}}_{mj}=\mathbf{h}_{mj}-\mathbf{\hat{h}}_{mj}$ as the channel estimation error, which is independent of $\mathbf{h}_{mj}$. The
elements of $\mathbf{\tilde{h}}_{mj}$ can be modeled as i.i.d. $\mathcal{CN}(0,\varepsilon_{mj})$ RVs, where $ \varepsilon_{mj}\in\{\varepsilon_{km}^{\dl},\varepsilon_{m\ell}^{\ul}\} $ corresponds to $ \beta_{mj} $, and 
\begin{IEEEeqnarray}{cl} \label{eq: MSE}
	\varepsilon_{mj}  =   \beta_{mj}  \Bigl(1 - \frac{\tau p_{j}^{\mathtt{tr}} \beta_{mj}}{\sum_{j'\in\mathcal{T}_{\ul}} \tau p_{j'}^{\mathtt{tr}} \beta_{mj'} |\boldsymbol{\bar{\varphi}}_{j}^H\boldsymbol{\bar{\varphi}}_{j'}|^2 + \sigma^2_{\mathtt{AP}}} \Bigr).\quad\
\end{IEEEeqnarray}
In an analogous fashion,  the channel estimate and  channel estimation error of CCI link $g_{k\ell}^{\mathtt{cci}}$ executed at $ \DLU$ are given as
\begin{align} \label{eq: channel estimate CCI}
\hat{g}_{k\ell}^{\mathtt{cci}}  = \frac{\sqrt{\tau p_{\ell}^{\mathtt{tr}} }\beta_{k\ell}^{\mathtt{cci}}}{\sum_{\ell'\in\mathcal{L}} \tau p_{\ell'}^{\mathtt{tr}} \beta_{k\ell'}^{\mathtt{cci}} |\boldsymbol{\bar{\varphi}}_{\ell}^H\boldsymbol{\bar{\varphi}}_{\ell'}|^2 + \sigma_{k}^2} \boldsymbol{\bar{\varphi}}_{\ell}^H\mathbf{y}_{k}^{\mathtt{tr},\mathtt{cci}}, 
\end{align}
and $\tilde{g}_{k\ell}^{\mathtt{cci}}\sim\mathcal{CN}(0,\varepsilon_{k\ell}^{\mathtt{cci}})$, respectively, where	
\begin{align}
\varepsilon_{k\ell}^{\mathtt{cci}}  = \beta_{k\ell}^{\mathtt{cci}}\Bigl(1-\frac{\tau p_{\ell}^{\mathtt{tr}} \beta_{k\ell}^{\mathtt{cci}}}{\sum_{\ell'\in\mathcal{L}} \tau p_{\ell'}^{\mathtt{tr}} \beta_{k\ell'}^{\mathtt{cci}} |\boldsymbol{\bar{\varphi}}_{\ell}^H\boldsymbol{\bar{\varphi}}_{\ell'}|^2 + \sigma_{k}^2}\Bigr).
\end{align}
Here, $\beta_{k\ell}^{\mathtt{cci}}$ denotes the large-scale fading of  CCI link  $\ULU \rightarrow\DLU$,  and   $ \mathbf{y}_{k}^{\mathtt{tr},\mathtt{cci}}=\sum_{\ell\in\mathcal{T}_{\ul}} \sqrt{\tau p_{j'}^{\mathtt{tr}}}\boldsymbol{\bar{\varphi}}_{\ell} g_{k\ell}^{\mathtt{cci}} + \mathbf{z}_{k} $, with $ \mathbf{z}_{k}\sim\mathcal{CN}(0,\sigma_{k}^2\mathbf{I}) $, is the UL UEs' training signals  received at $ \DLU $.

To mitigate the effects of pilot contamination, a  pilot assignment for the main DL and UL channels  is far more important that of CCI channels. Thus, we consider the following pilot assignment problem:
\begingroup\allowdisplaybreaks\begin{subequations} \label{eq: prob. MSE}
	\begin{IEEEeqnarray}{cl}
		\underset{\boldsymbol{\Upsilon}}{\min} &\quad  \underset{j\in\mathcal{T}_{\ul}}{\max} \sum\nolimits_{m\in\mathcal{M}} \frac{N_{m} \varepsilon_{mj}}{\beta_{mj}} \label{eq: prob. MSE :: a} \\
		\st & \quad \upsilon_{ij}\in\{0,1\},\;\sum\nolimits_{i\in\mathcal{T}_{\mathtt{p}} } \upsilon_{ij} \leq 1,\;\forall i\in\mathcal{T}_{\mathtt{p}},\forall j\in\mathcal{T}_{\ul}.\label{eq: prob. MSE :: b} \qquad
	\end{IEEEeqnarray}							
\end{subequations}\endgroup

\section{Proposed Heap-Based Pilot Assignment Strategy}
Problem \eqref{eq: prob. MSE} is a min-max problem for  sum of fractional functions, for which it is hard to find an optimal solution.  For an efficient solution, we first introduce the following theorem.
\begin{theorem} \label{thm: MSE prob.}
	Problem  \eqref{eq: prob. MSE} can be solved via the following tractable problem:
	\begin{IEEEeqnarray}{cl}\label{eq: prob. MSE quad.}
		\underset{\boldsymbol{\Upsilon}}{\min} \quad  \underset{j\in\mathcal{T}_{\ul}}{\max}  \sum\nolimits_{j'\in\mathcal{T}_{\ul}}  \tilde{\beta}_{j'}\boldsymbol{\upsilon}_{j}^H\boldsymbol{\upsilon}_{j'},\ \st\quad \eqref{eq: prob. MSE :: b},\quad
	\end{IEEEeqnarray}							
	where $ \tilde{\beta}_{j'}\triangleq\sum_{m\in\mathcal{M}}N_{m}\tau p_{j'}^{\mathtt{tr}} \beta_{mj'} $. 
\end{theorem}
\begin{proof}
	Please see \cite[Appendix F]{nguyen2019spectral}.
\end{proof}
We now propose the heap structure-based pilot assignment strategy. To do this,  the following definition is invoked.
\begin{definition}
	Min heap $ (\mathcal{H}^{\mathtt{min}}) $ is a tree-based structure, where  $ \mathcal{H}_{\mathtt{p}}^{\mathtt{min}} $ is a parent node of an arbitrary node $ \mathcal{H}_{\mathtt{c}}^{\mathtt{min}} $. Then, the key of $ \mathcal{H}_{\mathtt{p}}^{\mathtt{min}} $ is less than or equal to that of $ \mathcal{H}_{\mathtt{c}}^{\mathtt{min}} $. In a max heap $ (\mathcal{H}^{\mathtt{max}}) $, the key of $ \mathcal{H}_{\mathtt{p}}^{\mathtt{max}} $ is greater than or equal to that of $ \mathcal{H}_{\mathtt{c}}^{\mathtt{max}} $ \cite{Cormen:2001:IA:500824}.
\end{definition}

Let $ \mathcal{H}\in\{ \mathcal{H}^{\mathtt{min}}, \mathcal{H}^{\mathtt{max}} \} $, the main operations of heap structure include: $\bigl( \mathcal{G}(\mathbf{x},\{\mathbf{y}\}) \rightarrowtail \mathcal{H} \bigl)$ to \textit{generate a heap}, $\bigl( \mathcal{H} \rightarrow (x,\{\mathbf{y}\}) \bigl)$ to \textit{find min/max value}, $\bigl( \mathcal{H} \vdash (x,\{\mathbf{y}\}) \bigl)$ to \textit{extract the root node}, and $\bigl( \mathcal{H} \dashv (x,\{\mathbf{y}\}\bigr) $ to \textit{replace and sift-down} \cite{nguyen2019spectral}.
The proposed algorithm for pilot assignment is summarized in Algorithm~\ref{alg: MSE problem}. It takes the complexity of $ \mathcal{O}\bigr(U\log_{2}(\tau U )\bigl) $ for deriving the assignment solution, which has relatively low complexity.  For simplicity, this training strategy is referred to as Heap-FD, in which Algorithm~\ref{alg: MSE problem} is operated twice, i.e., with $ U=L $ for UL channel estimation in the first phase, and with $ U=K $ for achieving DL  and CCI channel estimates in the second phase. On the other hand, the training strategy for HD systems can be done by setting $ U=K+L $, called Heap-HD. 


\begin{algorithm}[t]
	\begin{algorithmic}[1]
				\fontsize{9}{10}\selectfont
		\protect\caption{Proposed Heap-Based Pilot Assignment for MSE Minimization Problem \eqref{eq: prob. MSE quad.}}
		
		\label{alg: MSE problem}
		
		\STATE Compute $ \boldsymbol{\tilde{\beta}}\triangleq[\tilde{\beta}_{j}]_{j\in\mathcal{T}_{\ul}} $ as in \textbf{Theorem}~\ref{thm: MSE prob.}.
		
		\STATE Randomly assign $ \tau $ pilots to the $ \tau $ first UEs in $ \mathcal{T}_{\ul} $, yielding $ \boldsymbol{\upsilon}_{j},\;\forall j=1,\cdots,\tau $.
		
		\STATE Execute $ \mathcal{G}([\boldsymbol{\tilde{\beta}}]_{1:\tau}, \{\boldsymbol{\upsilon}_j \}_{j=1,\cdots,\tau})\rightarrowtail\mathcal{H}^{\mathtt{min}} $.
		
		\STATE Execute $ \mathcal{G}([\boldsymbol{\tilde{\beta}}]_{\tau+1:U},\{\tau+1,\cdots,U\} )\rightarrowtail\mathcal{H}^{\mathtt{max}} $.
		
		\WHILE {$ \mathcal{H}^{\mathtt{max}}\neq \emptyset $}
		\STATE $ \mathcal{H}^{\mathtt{max}} \vdash (\tilde{\beta}_{j'},\{j'\}) $. \COMMENT{\textit{Root node is removed from} $ \mathcal{H}^{\mathtt{max}} $}
		\STATE $ \mathcal{H}^{\mathtt{min}} \rightarrow (\bar{\beta}_{i},\{\mathbf{\boldsymbol{\upsilon}}_{i}\}) $. 
		\STATE $ \mathbf{\boldsymbol{\upsilon}}_{j'}:=\mathbf{\boldsymbol{\upsilon}}_{i} $.
		\STATE $ \mathcal{H}^{\mathtt{min}} \dashv (\bar{\beta}_{i}+\tilde{\beta}_{j'},\{\mathbf{\boldsymbol{\upsilon}}_{i}\}) $.
		\ENDWHILE 
		
		\STATE Concatenate  assignment variable vectors as $ \boldsymbol{\Upsilon}:=[\boldsymbol{\upsilon}_{1},\cdots,\boldsymbol{\upsilon}_{U}] $. 
		
		
		\STATE {\textbf{Output:}  Pilot assignment matrix $\boldsymbol{\bar{\Xi}} = \boldsymbol{\Xi}\boldsymbol{\Upsilon}$}.
	\end{algorithmic} 
\end{algorithm}

\section{Optimization Problem for Robust Design} \label{OptimizationProblemDesign}
Let us denote by $ x_{k}^\dl $ and $ x_{\ell}^\ul$ the data symbols with unit power (i.e., $ \mathbb{E}\bigl[|x_{k}^{\dl}|^2\bigr]=1$ and $\mathbb{E}\bigl[|x_{\ell}^{\ul}|^2\bigr]=1$) intended for $\DLU$ and sent from $\ULU$, respectively. The beamforming vector $ \mathbf{w}_{km}\in \mathbb{C}^{N_m \times 1} $ is employed to precode the data symbol $ x_{k}^\dl $ of $\DLU$ in the DL, while $ p_{\ell}$ denotes the transmit power  of $\ULU$ in the UL. Let us introduce the new binary variables $ \alpha_{km}\in\{0,1\}, \forall k\in\mathcal{K}, m\in\mathcal{M}$ to  represent the association relationship between $\AP$ and $\DLU$, i.e.,
 $ \alpha_{km}=1 $ implying that $\DLU$ is served by $ \AP $ and $ \alpha_{km}=0 $, otherwise.  Using these notations, the  signal received at  $\DLU$ can be expressed as
\begingroup\allowdisplaybreaks\begin{align} \label{eq: received signal at DLU}
	y_{k}^\dl = & \sum_{m\in\mathcal{M}}\alpha_{km}\mathbf{h}_{km}^\dl\mathbf{w}_{km}x_{k}^\dl + \sum_{\ell\in\mathcal{L}} \sqrt{p_{\ell}} g_{k\ell}^{\mathtt{cci}} x_{\ell}^\ul\nonumber\\
	&\qquad + \sum_{m\in\mathcal{M}}\sum_{k'\in\mathcal{K}\setminus\{k\}}\alpha_{k'm}\mathbf{h}_{km}^\dl\mathbf{w}_{k'm}x_{k'}^\dl  + n_{k}, 
\end{align}\endgroup
where $ n_k\sim \mathcal{CN}(0,\sigma_{k}^2) $ is the additive white Gaussian noise (AWGN) and $\sigma_{k}^2$ is the noise variance. We consider a worst-case robust design by treating  CSI errors as noise. As a result, the received SINR at  $ \DLUi{k} $ is given as
\begin{align} \label{eq: DL SINR - general}
\gamma_{k}^{\dl}(\mathbf{w}, \mathbf{p},\boldsymbol{\alpha}) = \frac{\sum_{m\in\mathcal{M}}\alpha_{km}|\mathbf{\hat{h}}_{km}^\dl\mathbf{w}_{km}|^2 }{\chi_k(\mathbf{w}, \mathbf{p},\boldsymbol{\alpha})},
\end{align}
where $\chi_k(\mathbf{w}, \mathbf{p},\boldsymbol{\alpha})\triangleq\sum_{m\in\mathcal{M}}\sum_{k'\in\mathcal{K}\setminus\{k\}}\alpha_{k'm}|\mathbf{h}_{km}^\dl\mathbf{w}_{k'm}|^2 + \sum_{m\in\mathcal{M}}\sum_{k'\in\mathcal{K}}\alpha_{k'm}\varepsilon_{k'm}^{\dl}\|\mathbf{w}_{k'm}\|^2+\sum_{\ell\in\mathcal{L}}p_{\ell}|\hat{g}_{k\ell}^{\mathtt{cci}}|^2+\sum_{\ell\in\mathcal{L}}p_{\ell}\varepsilon_{k\ell}^{\mathtt{cci}}+\sigma_{k}^2$, $\mathbf{w}\triangleq[\mathbf{w}_{1}^H,\cdots,\mathbf{w}_K^H]^H\in\mathbb{C}^{NK\times 1}$ with $\mathbf{w}_k\triangleq[\mathbf{w}_{k1}^H,\cdots,\mathbf{w}_{kM}^H]^H\in\mathbb{C}^{N\times 1}$, $\mathbf{p}\triangleq[p_{1},\cdots,p_L]^T\in\mathbb{R}^{L\times 1}$, and $\boldsymbol{\alpha}\triangleq\{\alpha_{km}\}_{\forall k\in\mathcal{K}, m\in\mathcal{M}}$. We note that in \eqref{eq: DL SINR - general}, $ \alpha_{km} $ is equal to $ \alpha_{km}^2 $ for any $\alpha_{km}\in\{0,1\}$.

\vspace{-1pt}
The received signal at $ \AP $ can be expressed as
\begingroup
\allowdisplaybreaks{\small\begin{IEEEeqnarray}{rCl} \label{eq: received signal at AP}
\mathbf{y}_m^{\ul} & = \sum_{\ell\in\cL} \sqrt{p_{\ell}}  \mathbf{h}_{m\ell}^{\ul} x_{\ell}^{\ul}  + \sum_{ \substack{m'\in\cM}}\sum_{k\in\cK} \alpha_{km'} \mathbf{G}_{mm'}^{\AtoA} \mathbf{w}_{km'} x_{k}^{\dl} + \mathbf{n}_m,\quad
\end{IEEEeqnarray}}\endgroup
 where $ \mathbf{n}_m\sim \mathcal{CN}(0,\sigma_{\mathtt{AP}}^2\mathbf{I}) $ is the AWGN.  To decode the $ \ULU $'s message, let us denote the receiver vector by $\mathbf{a}_{m\ell}\in\mathbb{C}^{1\times N_m}$, and thus, the received signal of $ \ULU $ at $ \APx{m} $ can be expressed as $r_{m\ell}^{\ul}=\mathbf{a}_{m\ell}\mathbf{y}_m^{\ul}$.
Consequently, the post-detection signal for decoding the $ \ULU $'s signal is $r_{\ell}^{\ul}=\sum_{m\in\mathcal{M}}r_{m\ell}^{\ul}$.  By defining $\mathbf{\hat{h}}_{\ell}^{\ul}\triangleq\bigl[(\mathbf{\hat{h}}_{1\ell}^{\ul})^H,\cdots,(\mathbf{\hat{h}}_{M\ell}^{\ul})^H\bigr]^H\in\mathbb{C}^{N\times 1}$, $ \mathbf{\bar{G}}_{m'}^{\AtoA}\triangleq\bigl[(\mathbf{G}_{1m'}^{\AtoA})^H,\cdots,(\mathbf{G}_{Mm'}^{\AtoA})^H\bigr]^H\in\mathbb{C}^{N\times N_{m'}}$,  $ \mathbf{a}_{\ell}=[\mathbf{a}_{1\ell},\cdots,\mathbf{a}_{M\ell}] \in\mathbb{C}^{1\times N}$ and  $\mathbf{n}\triangleq[\mathbf{n}_1^H,\cdots,\mathbf{n}_M^H]^H\in\mathbb{C}^{N\times 1}$, the SINR in decoding   $ \ULU $'s message is given as
\begin{IEEEeqnarray}{rCl} \label{eq: UL SINR - general}
	\gamma_{\ell}^{\ul}(\mathbf{w}, \mathbf{p},\boldsymbol{\alpha}) = \frac{ p_{\ell}|\mathbf{a}_{\ell}  \mathbf{\hat{h}}_{\ell}^{\ul}|^2 }{ \mathcal{I}^{\AtoA}_{\ell} },\quad
\end{IEEEeqnarray}
where $ \mathcal{I}^{\AtoA}_{\ell}\triangleq \sum_{\ell'\in\mathcal{L}\setminus\{\ell\}} p_{\ell'} |\mathbf{a}_{\ell}  \mathbf{\hat{h}}_{\ell'}^{\ul}|^2+ \sum_{\ell'\in\mathcal{L}} p_{\ell'} \varepsilon_{m\ell'}^{\ul} \|\mathbf{a}_{\ell}\|^2+\sum_{m'\in\mathcal{M}}\sum_{k\in\mathcal{K}} \alpha_{km'} |\mathbf{a}_{\ell} \mathbf{\bar{G}}_{m'}^{\AtoA}\mathbf{w}_{km'}|^2+ \sigma_{\mathtt{AP}}^2\|\mathbf{a}_{\ell}\|^2 $.
\subsection{General Problem Formulation}
From \eqref{eq: DL SINR - general} and \eqref{eq: UL SINR - general}, the SE is given as
$	F_{\mathtt{SE}}\bigl(\mathbf{w}, \mathbf{p},\boldsymbol{\alpha}\bigr)\triangleq \sum\nolimits_{k\in\mathcal{K}}R\bigl(\gamma_{k}^{\dl}(\mathbf{w}, \mathbf{p},\boldsymbol{\alpha})\bigr)
	+\sum\nolimits_{\ell\in\mathcal{L}}R\bigl(\gamma_{\ell}^{\ul}(\mathbf{w}, \mathbf{p},\boldsymbol{\alpha})\bigr),
$
where $R(x)\triangleq \ln(1+x)$. Then, a joint design of  power control and AP-DL UE association is formulated as
\begingroup
\allowdisplaybreaks\begin{subequations} \label{eq: prob. general form bi-obj. trade-off}
	\begin{IEEEeqnarray}{cl}
		\underset{\mathbf{w}, \mathbf{p},\boldsymbol{\alpha}}{\max} &\quad  F_{\mathtt{SE}}\bigl(\mathbf{w}, \mathbf{p},\boldsymbol{\alpha}\bigr)  \label{eq: prob. general form bi-obj. trade-off :: a} \\
		\st & \quad \alpha_{km} \in \{0,1\}, \; \forall k \in \mathcal{K},m \in \mathcal{M}, \label{eq: prob. general form bi-obj. trade-off :: g} \\
		& \quad \|\mathbf{w}_{km}\|^2 \leq \alpha_{km} P_{\mathtt{AP}_m}^{\max} , \; \forall k \in \mathcal{K},m \in \mathcal{M}, \label{eq: prob. general form bi-obj. trade-off :: h} \\
		&\quad \sum\nolimits_{k\in\cK}\alpha_{km}\|\mathbf{w}_{km}\|^2 \leq  P_{\mathtt{AP}_m}^{\max},\;\forall m\in\mathcal{M}, \label{eq: prob. general form bi-obj. trade-off :: b} \qquad\\
		& \quad 0 \leq p_{\ell} \leq P_{\ell}^{\max},\; \forall \ell \in \mathcal{L}, \label{eq: prob. general form bi-obj. trade-off :: c} \\
		& \quad 
		R\bigl(\gamma_{k}^{\dl}(\mathbf{w}, \mathbf{p},\boldsymbol{\alpha})\bigr) \geq \bar{R}_{k}^{\dl}, \; \forall  k \in \mathcal{K}, \label{eq: prob. general form bi-obj. trade-off :: d} \\
		&  \quad
		R\bigl(\gamma_{\ell}^{\ul}(\mathbf{w}, \mathbf{p},\boldsymbol{\alpha})\bigr) \geq \bar{R}_{\ell}^{\ul}, \; \forall \ell \in \mathcal{L}. \label{eq: prob. general form bi-obj. trade-off :: e}
	\end{IEEEeqnarray}							
\end{subequations}\endgroup
 Constraint \eqref{eq: prob. general form bi-obj. trade-off :: h} is used for the user selection, while constraints \eqref{eq: prob. general form bi-obj. trade-off :: b} and \eqref{eq: prob. general form bi-obj. trade-off :: c} imply that the  transmit powers at  $\AP$ and $\DLU$ are limited by their maximum power budgets $ P_{\mathtt{AP}_m}^{\max} $ and $ P_{\ell}^{\max} $, respectively.  Moreover, constraints \eqref{eq: prob. general form bi-obj. trade-off :: d} and \eqref{eq: prob. general form bi-obj. trade-off :: e} are used to ensure the predetermined rate requirements $\bar{R}_{k}^{\dl}$ and $\bar{R}_{k}^{\ul}$ for $\DLU$ and $\ULU$, respectively. We can see that problem \eqref{eq: prob. general form bi-obj. trade-off} is a mixed-integer nonconvex optimization problem.

\subsection{Tractable Problem Formulation  for \eqref{eq: prob. general form bi-obj. trade-off}}
For solving  problem \eqref{eq: prob. general form bi-obj. trade-off}, it is not practical to try all possible AP-DL UE associations, especially for networks of large size. To overcome this issue, we exploit the special relationship between  continuous   and  binary variables \cite[\textbf{Lemma} 1 and \textbf{Theorem} 1]{nguyen2019spectral}. Particularly, we define $\boldsymbol{\Gamma}_{\dl}\triangleq\bigl\{\gamma_{k}^{\dl}(\mathcal{C},\mathbf{1})|\forall k\in\mathcal{K}\bigr\}$ and $\boldsymbol{\Gamma}_{\ul}\triangleq\bigl\{\gamma_{\ell}^{\ul}(\mathcal{C},\mathbf{1})|\forall \ell\in\mathcal{L}\bigr\}$ with all entries of $ \boldsymbol{\alpha} $ being replaced by ones. In short, problem \eqref{eq: prob. general form bi-obj. trade-off} can be rewritten as
\begingroup
\allowdisplaybreaks
\begin{subequations} \label{eq: prob. bi-obj. - no alpha}
	\begin{IEEEeqnarray}{cl}
		\underset{\mathcal{C}\triangleq\{\mathbf{w}, \mathbf{p}\}}{\max} &\quad \bar{F}_{\mathtt{SE}}\bigl(\boldsymbol{\Gamma}_{\dl},\boldsymbol{\Gamma}_{\ul}\bigr)  \label{eq: prob. bi-obj. - no alpha :: a} \qquad\\
		\st & \quad  \sum\nolimits_{k\in\cK}\|\mathbf{w}_{km}\|^2 \leq  P_{\mathtt{AP}_m}^{\max},\;\forall m\in\mathcal{M}, \label{eq: prob. bi-obj. - no alpha :: c} \\
				& \quad 
		R\bigl(\gamma_{k}^{\dl}(\mathcal{C},\mathbf{1})\bigr) \geq \bar{R}_{k}^{\dl}, \; \forall  k \in \mathcal{K}, \label{eq: prob. bi-obj. - no alpha :: e} \\
		&  \quad
		R\bigl(\gamma_{\ell}^{\ul}(\mathcal{C},\mathbf{1})\bigr) \geq \bar{R}_{\ell}^{\ul}, \; \forall \ell \in \mathcal{L}, \label{eq: prob. bi-obj. - no alpha :: f}\\
		&  \quad \eqref{eq: prob. general form bi-obj. trade-off :: c},\label{eq: prob. bi-obj. - no alpha :: g}
	\end{IEEEeqnarray}							
\end{subequations}
\endgroup
where $\bar{F}_{\mathtt{SE}}\bigl(\boldsymbol{\Gamma}_{\dl},\boldsymbol{\Gamma}_{\ul}\bigr)  \triangleq  R_{\Sigma}(\boldsymbol{\Gamma}_{\dl}) + R_{\Sigma}(\boldsymbol{\Gamma}_{\ul}),$
with $ R_{\Sigma}(\mathcal{X}) = \sum_{x\in\mathcal{X}} R(x) $. The signal-power ratio function is defined as
\begin{IEEEeqnarray}{cl}
f_{\mathtt{spr}}:\mathbf{w}\rightarrow \mathbf{r}^{\mathtt{sp}}\triangleq\bigl[r_{\mathtt{sp}}\bigl(\mathbf{w}_{km},\mathbf{\hat{h}}_{km}^{\dl}|\mathbf{w}_{k}^{(\kappa)},\mathbf{\hat{h}}_{k}^{\dl}\bigr)\bigr]_{\forall k\in\mathcal{K},m\in\mathcal{M}},\qquad
\end{IEEEeqnarray}
with $\mathbf{\hat{h}}_k^\dl\triangleq[\mathbf{\hat{h}}_{k1}^\dl,\cdots,\mathbf{\hat{h}}_{kM}^\dl]\in\mathbb{C}^{1\times N}$, and
\begin{align} \label{eq: consumed power rate function}
	r_{\mathtt{sp}}\bigl(\mathbf{x}_{1},\mathbf{c}_{1}|\mathbf{x}_{2},\mathbf{c}_{2}\bigr)\triangleq\frac{|\mathbf{c}_{1}\mathbf{x}_{1}|^2}{|\mathbf{c}_{2}\mathbf{x}_{2}|^2+\epsilon} \in [0,1) ,
\end{align}
where $\epsilon $ is a very small real number added to avoid a numerical problem, i.e., $ 10^{-6} $, and $ \mathbf{w}_{k}^{(\kappa)}$ is a feasible point of $ \mathbf{w}_{k}$ at the $\kappa$-th iteration of an iterative algorithm presented shortly. We can obtain $ \boldsymbol{\alpha} $ via a converter function, i.e., for $ \forall k\in\mathcal{K},m\in\mathcal{M} $,
\begin{align}\label{eq: Bfunction}
\alpha_{km}^* = \begin{cases}
1, & \text{if } r_{\mathtt{sp}}\bigl(\mathbf{w}_{km}^{*},\mathbf{\hat{h}}_{km}^{\dl}|\mathbf{w}_{k}^{*},\mathbf{\hat{h}}_{k}^{\dl}\bigr)>\varpi, \\
0, & \text{if } r_{\mathtt{sp}}\bigl(\mathbf{w}_{km}^{*},\mathbf{\hat{h}}_{km}^{\dl}|\mathbf{w}_{k}^{*},\mathbf{\hat{h}}_{k}^{\dl}\bigr)\leq\varpi,
\end{cases}
\end{align}
and the per-AP power signal ratio $\varpi\triangleq10^{-3}/M$ is a small number, and $ \mathbf{w}_{km}^{*} $ is the optimal solution of $ \mathbf{w}_{km}$. It is true that $ r_{\mathtt{sp}}\bigl(\mathbf{w}_{km}^{*},\mathbf{\hat{h}}_{km}^{\dl}|\mathbf{w}_{k}^{*},\mathbf{\hat{h}}_{k}^{\dl}\bigr)\leq\varpi $ yields  $\mathbf{w}_{km}^{*} \rightarrow \mathbf{0} $. Without loss of optimality, we can omit $ \boldsymbol{\alpha} $ in the following derivations.

\section{Proposed Solution Based on Zero-Forcing}\label{sec:ZF}

\subsection{ZF-Based Transmission Design}
For ease of presentation, we first let $\mathbf{W}\triangleq[\mathbf{w}_{1},\cdots,\mathbf{w}_{K}]\in\mathbb{C}^{N\times K}$, $\mathbf{H}^\dl\triangleq\bigl[(\mathbf{\hat{h}}_{1}^\dl)^H, \cdots,(\mathbf{\hat{h}}_{M}^\dl)^H\bigr]^H \in \mathbb{C}^{K\times N}$,  $ \mathbf{H}^{\ul}\triangleq\bigl[\mathbf{\hat{h}}_{1}^{\ul},\cdots, \mathbf{\hat{h}}_{L}^{\ul} \bigr]\in\mathbb{C}^{N\times L} $, $\mathbf{G}^{\mathtt{cci}}\triangleq\bigl[(\mathbf{g}_{1}^{\mathtt{cci}})^H,\cdots, (\mathbf{g}_{K}^{\mathtt{cci}})^H\bigr]^H\in \mathbb{C}^{K\times L}$ with $\mathbf{g}_{k}^{\mathtt{cci}}\triangleq\bigl[\hat{g}_{k1}^{\mathtt{cci}},\cdots,\hat{g}_{kL}^{\mathtt{cci}}\bigr]\in \mathbb{C}^{1\times L}$,  $ \mathbf{\tilde{G}}^{\AtoA}\triangleq\bigl[\mathbf{\bar{G}}_1^{\AtoA},\cdots,$ $\mathbf{\bar{G}}_M^{\AtoA}\bigr]\in\mathbb{C}^{N\times N} $,  and $\mathbf{D}^{\ul}\triangleq\diag\bigl(\bigl[\sqrt{p_{1}}\cdots \sqrt{p_{L}}\bigr]\bigr)$.

\subsubsection{ZF-Based  DL Transmission}
For $ \mathbf{H}^\tZF= (\mathbf{H}^\dl)^H\bigl(\mathbf{H}^\dl(\mathbf{H}^\dl)^H\bigr)^{-1}$, the ZF precoder matrix is simply computed as
$\mathbf{W}=\mathbf{W}^{\mathtt{ZF}} = \mathbf{H}^\tZF (\mathbf{D}^{\dl})^{\frac{1}{2}},
$
where  $ \mathbf{D}^{\dl}=\diag\bigl(\bigl[\omega_1\cdots\omega_K\bigr]\bigr) $ and $ \omega_{k} $ represents the  weight for $\DLU$. As a result, constraint  \eqref{eq: prob. bi-obj. - no alpha :: c} becomes
\begin{align} \label{eq: power constraint - ZF}
\tr\bigl((\mathbf{H}^{\tZF})^H\mathbf{B}_{m}\mathbf{H}^{\tZF}\mathbf{D}^{\dl}\bigr) \leq  P_{\mathtt{AP}_m}^{\max}, \;\forall m\in\mathcal{M},
\end{align}
where $ \mathbf{B}_{m} = \diag(\mathbf{b}_{m})\in \{0,1\}^{N\times N} $ with 
\begin{align}
\mathbf{b}_{m} = \bigl( \underbrace{0\cdots 0}_{\sum\nolimits_{m'=1}^{m-1}N_{m'}}\;\underbrace{1\cdots 1}_{N_m}\; 0 \cdots 0\bigr).
\end{align}
The simplicity of ZF is attributed  to the fact that the size of $NK$ scalar variables of $\mathbf{w}$ is now reduced to $K$ scalar variables of $\boldsymbol{\omega}\triangleq[\omega_1,\cdots,\omega_K]^T\in\mathbb{R}^{K\times 1}$. The SINR of $\DLU$ with ZF precoder is
\begin{align}\label{eq:ZFSINIDL}
\gamma_{k}^{\dl,\tZF}(\boldsymbol{\omega}, \mathbf{p}) = \frac{\omega_k|\mathbf{\hat{h}}_{k}^\dl\mathbf{h}_{k}^\tZF|^2}{\mathcal{I}_{k}^{\mathtt{err},\dl}+ \|\mathbf{g}_{k}^{\mathtt{cci}}\mathbf{D}^{\ul}\|^2+\sigma_{k}^2},
\end{align}
where $ \mathcal{I}_{k}^{\mathtt{err},\dl}\triangleq\sum_{m\in\mathcal{M}}\sum_{k'\in\mathcal{K}}\varepsilon_{k'm}^{\dl}\omega_{k'}|\mathbf{b}_{m}\mathbf{h}_{k'}^\tZF|^2+\sum_{\ell\in\mathcal{L}}p_{\ell}\varepsilon_{k\ell}^{\mathtt{cci}} $; $ \mathbf{h}_{k}^\tZF $ is the $ k $-th column of the ZF precoder $ \mathbf{H}^{\tZF} $ and the MUI term $ |\mathbf{\hat{h}}_{k}^\dl\mathbf{w}_{k'}|^2 \approx 0, \forall k'\in\cK\setminus\{k\}$. The relationship between  $ \boldsymbol{\omega} $ and $ \mathbf{W}^{\mathtt{ZF}} $ is characterized as
\begin{align}\label{eq:relationshipWOmega}
\mathbf{W}^{\mathtt{ZF}}=\mathbf{H}^{\tZF}\bigl(\diag(\boldsymbol{\omega})\bigr)^{\frac{1}{2}}.
\end{align}
Hence,  $ \mathbf{w}_{km} $ is recovered by extracting from the $ ((m-1)N_m+1) $-th to $ (mN_m) $-th elements of $ \mathbf{w}_k $ - the $ k $-th column of $ \mathbf{W}^{\mathtt{ZF}} $.


\subsubsection{ZF-Based  UL Transmission}
Let $ \mathbf{A}^{\mathtt{ZF}}=\bigl((\mathbf{H}^\ul)^H\mathbf{H}^\ul\bigr)^{-1}(\mathbf{H}^\ul)^H \in\mathbb{C}^{L\times N} $ be the ZF receiver matrix at the CPU. The SINR of $ \ULU $ is
\begin{align}\label{eq:ZFSINIUL}
\gamma_{\ell}^{\ul,\tZF}(\boldsymbol{\omega},\mathbf{p}) = \frac{p_{\ell}|\mathbf{a}_{\ell}^{\mathtt{ZF}}\mathbf{\hat{h}}_{\ell}^{\ul}|^2 }{\mathcal{I}_{\ell}^{\mathtt{err},\ul}+\|\mathbf{a}_{\ell}^{\mathtt{ZF}}\mathbf{\tilde{G}}^{\AtoA} \mathbf{W}^{\mathtt{ZF}}\|^2+\sigma_{\mathtt{AP}}^2\|\mathbf{a}_{\ell}^{\mathtt{ZF}}\|^2},
\end{align}
where $ \mathbf{a}_{\ell}^{\mathtt{ZF}} $ is the $\ell$-th row of $ \mathbf{A}^{\mathtt{ZF}}$, and $ \mathcal{I}_{\ell}^{\mathtt{err},\ul}\triangleq\sum_{\ell'\in\mathcal{L}} p_{\ell'} \varepsilon_{m\ell'}^{\ul} \|\mathbf{a}_{\ell}^{\mathtt{ZF}}\|^2 $.

\subsection{Proposed Algorithm}
Before proceeding, we provide some useful approximate functions as follows.
\begin{align} 
h_{\mathtt{fr}}(x,y) &\triangleq\frac{x^2}{y} \geq \frac{2x^{(\kappa)}}{y^{(\kappa)}}x - \frac{(x^{(\kappa)})^2}{(y^{(\kappa)})^2}y := h_{\mathtt{fr}}^{(\kappa)}(x,y), \\
h_{\mathtt{qu}}(x)&\triangleq x^2 \geq 2x^{(\kappa)} x - (x^{(\kappa)})^2 := h_{\mathtt{qu}}^{(\kappa)}(x). 
\end{align}

As in \cite[Theorem 3]{nguyen2019spectral}, problem \eqref{eq: prob. bi-obj. - no alpha} based on ZF is rewritten as the following problem
\begingroup
\allowdisplaybreaks\begin{subequations} \label{eq: prob. bi-obj. - equiZF}
	\begin{IEEEeqnarray}{cl}
		\underset{\boldsymbol{\omega}, \mathbf{p},\boldsymbol{\lambda}}{\max} &\quad  \tilde{F}_{\mathtt{SE}}\bigl(\boldsymbol{\Lambda}_{\dl},\boldsymbol{\Lambda}_{\ul}\bigr)		\label{eq: prob. bi-obj. - equiZF :: a} \\
		\st	
		& \quad \gamma_{k}^{\dl,\tZF}(\boldsymbol{\omega}, \mathbf{p})\geq\lambda_{k}^{\dl}, \; \forall  k \in \mathcal{K},\label{eq: prob. bi-obj. - equiZF :: c} \\
		& \quad \gamma_{\ell}^{\ul,\tZF}(\boldsymbol{\omega}, \mathbf{p})\geq\lambda_{\ell}^{\ul}, \; \forall  \ell \in \mathcal{L}, \label{eq: prob. bi-obj. - equiZF :: d} \\
		& \quad \lambda_{k}^{\dl} + 1 \geq \exp(\bar{R}_{k}^{\dl}), \; \forall  k \in \mathcal{K},   \label{eq: prob. bi-obj. - equiZF :: e} \\
		& \quad \lambda_{\ell}^{\ul} + 1 \geq \exp(\bar{R}_{\ell}^{\ul}), \; \forall  \ell \in \mathcal{L}, \label{eq: prob. bi-obj. - equiZF :: f} \\
		&\quad \eqref{eq: prob. general form bi-obj. trade-off :: c}, \eqref{eq: power constraint - ZF}, \label{eq: prob. bi-obj. - equiZF :: g} 
	\end{IEEEeqnarray}							
\end{subequations}\endgroup
where  $\tilde{F}_{\mathtt{SE}}\bigl(\boldsymbol{\Lambda}_{\dl},\boldsymbol{\Lambda}_{\ul}\bigr)\triangleq \ln|\mathbf{I} + \boldsymbol{\Lambda}_{\dl} | + \ln|\mathbf{I} + \boldsymbol{\Lambda}_{\ul} |$ is a concave function, with $ \boldsymbol{\Lambda}_{\dl}\triangleq\diag([\lambda_{1}^{\dl}\cdots\lambda_{K}^{\dl}]) $ and $ \boldsymbol{\Lambda}_{\ul}\triangleq\diag([\lambda_{1}^{\ul}\cdots\lambda_{L}^{\ul}]) $; $\boldsymbol{\lambda}\triangleq\{\boldsymbol{\lambda}_\dl,\boldsymbol{\lambda}_\ul\}$ with $ \boldsymbol{\lambda}_{\dl}\triangleq\{\lambda_{k}^{\dl}\}_{\forall k\in\mathcal{K}} $ and $ \boldsymbol{\lambda}_{\ul}\triangleq\{\lambda_{\ell}^{\dl}\}_{\forall \ell\in\mathcal{L}} $ is newly introduced variables.

In problem \eqref{eq: prob. bi-obj. - equiZF}, the nonconvex parts include \eqref{eq: prob. bi-obj. - equiZF :: c} and \eqref{eq: prob. bi-obj. - equiZF :: d}. We introduce the new variables as $\boldsymbol{\psi^{\dl}}\triangleq\{\psi_{k}^{\dl}\}_{\forall k\in\cK}$ and $\boldsymbol{\psi^{\ul}}\triangleq\{\psi_{\ell}^{\ul}\}_{\forall \ell\in\cL}$.
Based on the IA method, constraint \eqref{eq: prob. bi-obj. - equiZF :: c} is convexified as the two following linear constraints:
\begin{subequations} \label{eq: SINR cons DL convex.}
\begin{align}
	& h_{\mathtt{fr}}^{(\kappa)}(\sqrt{\omega_k},\psi_{k}^{\dl})  \geq \lambda_{k}^{\dl},\;\forall k\in\mathcal{K}, \\
	& \psi_{k}^{\dl}  \geq \frac{\mathcal{I}_{k}^{\mathtt{err},\dl}+\|\mathbf{g}_{k}^{\mathtt{cci}}\mathbf{D}^{\ul}\|^2+\sigma_{k}^2}{|\mathbf{\hat{h}}_{k}^\dl\mathbf{h}_{k}^{\tZF}|^2}, \forall k\in\mathcal{K},
\end{align}
\end{subequations}
while constraint \eqref{eq: prob. bi-obj. - equiZF :: d} is iteratively replaced by
\begin{subequations} \label{eq: prob. bi-obj. - equiZF :: dConcex}
	\begin{align} 
	&h_{\mathtt{fr}}^{(\kappa)}(\sqrt{p_{\ell}},\psi_{\ell}^{\ul})  \geq \lambda_{\ell}^{\ul},\;\forall \ell\in\mathcal{L}, \label{eq: prob. bi-obj. - equiZF :: dConcex:a}\\
	&\psi_{\ell}^{\ul} \geq \frac{\mathcal{I}_{\ell}^{\mathtt{err},\ul}+\|\mathbf{a}_{\ell}^{\mathtt{ZF}}\mathbf{\tilde{G}}^{\AtoA} \mathbf{W}^{\mathtt{ZF}}\|^2+\sigma_{\mathtt{AP}}^2\|\mathbf{a}_{\ell}^{\mathtt{ZF}}\|^2}{|\mathbf{a}_{\ell}^{\mathtt{ZF}}\mathbf{\hat{h}}_{\ell}^{\ul}|^2}, \forall \ell\in\mathcal{L}.\label{eq: prob. bi-obj. - equiZF :: dConcex:b}
	\end{align}
\end{subequations}

With the above discussions based on the IA method,  we obtain the following approximate problem of \eqref{eq: prob. bi-obj. - equiZF} with the convex set solved  at iteration $ (\kappa+1) $:
\begin{subequations} \label{eq: prob. bi-obj. - frac. prog.}
	\begin{IEEEeqnarray}{cl}
		\underset{\boldsymbol{\omega}, \mathbf{p},\boldsymbol{\lambda},\boldsymbol{\psi}}{\max} &\quad  \ddot{F}^{(\kappa+1)}\triangleq\tilde{F}_{\mathtt{SE}}\bigl(\boldsymbol{\Lambda}_{\dl},\boldsymbol{\Lambda}_{\ul}\bigr) \label{eq: prob. bi-obj. - frac. prog. :: a} \\
		\st&  \quad \eqref{eq: prob. general form bi-obj. trade-off :: c}, \eqref{eq: power constraint - ZF}, \eqref{eq: prob. bi-obj. - equiZF :: e}, \eqref{eq: prob. bi-obj. - equiZF :: f}, \eqref{eq: SINR cons DL convex.}, \eqref{eq: prob. bi-obj. - equiZF :: dConcex},
		 \label{eq: prob. bi-obj. - frac. prog. :: b} \qquad
	\end{IEEEeqnarray}							
\end{subequations}
where $\boldsymbol{\psi}\triangleq\{\boldsymbol{\psi^{\dl}},\boldsymbol{\psi^{\ul}}\}$. Clearly, the set of variables in \eqref{eq: prob. bi-obj. - frac. prog.} is independent of the numbers of APs (antennas) and all constraints are linear.
The proposed algorithm for solving the ZF-based SE problem \eqref{eq: prob. general form bi-obj. trade-off} is summarized in Algorithm \ref{alg: ZFD problem}, where the convergence and complexity analysis are given in \cite{nguyen2019spectral}.

\begin{algorithm}[t]
	\begin{algorithmic}[1]
				\fontsize{9}{10}\selectfont
		\protect\caption{Proposed Algorithm to Solve ZF-based SE Problem \eqref{eq: prob. general form bi-obj. trade-off}}
		
		\label{alg: ZFD problem}
		
		\STATE \textbf{Initialization:} Compute matrices $\mathbf{H}^{\tZF}$ and $\mathbf{A}^{\tZF}$. 
		\STATE Set $ \ddot{F}^{(0)}=-\infty $, $\kappa:=0$, and  generate $(\boldsymbol{\omega}^{(0)}, \mathbf{p}^{(0)},\boldsymbol{\psi}^{(0)})$.
		
		\REPEAT[Solving \eqref{eq: prob. bi-obj. - equiZF}]
		\STATE Solve \eqref{eq: prob. bi-obj. - frac. prog.} to obtain $(\boldsymbol{\omega}^{\star}, \mathbf{p}^{\star},\boldsymbol{\lambda}^{\star},\boldsymbol{\psi}^{\star})$ and $ \ddot{F}^{(\kappa+1)} $.
		
		\STATE Update $(\boldsymbol{\omega}^{(\kappa+1)}, \mathbf{p}^{(\kappa+1)},\boldsymbol{\psi}^{(\kappa+1)}) :=(\boldsymbol{\omega}^{\star}, \mathbf{p}^{\star},\boldsymbol{\psi}^{\star})$.
		\STATE Set $ \kappa := \kappa + 1 $.
		\UNTIL Convergence \COMMENT{$ \ddot{F}^{(\kappa)} - \ddot{F}^{(\kappa-1)} < 10^{-3} $}
		
		\STATE Update $(\boldsymbol{\omega}^{*}, \mathbf{p}^{*}) := (\boldsymbol{\omega}^{(\kappa)}, \mathbf{p}^{(\kappa)} )$.
		
		\STATE Use \eqref{eq:relationshipWOmega} to recover $\mathbf{w}_{km},\;\forall k\in\cK, m\in\cM$.
		
		\STATE Compute $ \boldsymbol{\alpha}^{*} $ as in \eqref{eq: Bfunction}.
		
		\STATE Repeat Steps 1-9 with $ \boldsymbol{\alpha}^{*} $ to find the exact $(\mathbf{w}^{*}, \mathbf{p}^{*})$.
		
		\STATE Use $ (\mathbf{w}^{*}, \mathbf{p}^{*},\boldsymbol{\alpha}^{*}) $ to compute $ F_{\mathtt{SE}}(\mathbf{w}^{*}, \mathbf{p}^{*},\boldsymbol{\alpha}^{*}) $ in \eqref{eq: prob. general form bi-obj. trade-off :: a}.

	\end{algorithmic} 
\end{algorithm}

%
%



\section{Numerical Results}\label{NumericalResults}

\begin{table}[!t]
	\centering
	\captionof{table}{Simulation Parameters}
	\label{tab: parameter}
	\vspace{-5pt}
	\scalebox{0.8}{
		\begin{tabular}{l|l}
			\hline
			Parameter & Value \\
			\hline\hline
			System bandwidth, $B$ & 10 MHz \\
			Residual SiS, $\rho^{\RSI}=\rho_{mm}^{\RSI},\;\forall m$& -110 dB \cite{Korpi:IEEETAP:Feb2017}\\
			Noise power at receivers & -104 dBm \\
			Number of APs and UEs, $ M $ & 64\\
			Number of antennas per AP, $ N_m,\forall m $ & 2\\
			Rate threshold, $  \bar{R}=\bar{R}_{k}^{\dl}=\bar{R}_{\ell}^{\ul} ,\;\forall k, \ell $	&  0.5 bits/s/Hz\\
			{\color{black}Power budget at UL UEs, $P_{\ell}^{\max},\forall\ell $} & 23 dBm \\
			{\color{black}Total power budget for all APs, $MP_{\mathtt{AP}}^{\max}$}  & 43 dBm \\
			\hline		   				
		\end{tabular}
	}
\end{table}

We consider a system topology with all  APs and UEs located within a circle of 1-km radius. The entries of the fading loop channel $ \mathbf{G}^{\SI}_{mm},\forall m\in\mathcal{M} $ are modeled as i.i.d. Rician RVs, with the Rician factor of $ 5 $ dB \cite{Dinh:JSAC:18}. The large-scale fading of other channels is modeled as $ \beta = 10^{\frac{\mathtt{PL}(d)+\sigma_{\mathtt{sh}}z }{10}} $, where $ \beta\in\{ \beta_{mm'}^{\AtoA},\beta_{km}^{\dl}, \beta_{m\ell}^{\ul},\beta_{k\ell}^{\mathtt{cci}} \} $, $\forall  m,m'\in\mathcal{M},k\in\mathcal{K},\ell\in\mathcal{L} $ and $ m\neq m' $; The shadow fading is considered as an RV $ z\in\{ z_{mm'}^{\AtoA},z_{km}^{\dl},z_{m\ell}^{\ul},z_{k\ell}^{\mathtt{cci}} \} \sim \mathcal{N}(0,1) $ with standard deviation $ \sigma_{\mathtt{sh}} =8$  dB. The three-slope model for the
path loss in dB is given by \cite{Ngo:TWC:Mar2017}.
The other parameters are given in Table~\ref{tab: parameter}, where all APs have the same power budget $P_{\mathtt{AP}}^{\max}=P_{\mathtt{AP}_m}^{\max},\forall m$. The SEs are divided by $ \ln2 $ to be presented in bits/s/Hz.

For comparison, the heap-/random-based training schemes and ZF-based robust design under FD operation (\textit{``Heap-FD + ZF-RD''} and \textit{``Rand-FD + ZF-RD''}) are employed to evaluate the  performance of those schemes, referenced to two approaches of Heap-HD (\textit{``Heap-HD + ZF-RD''}) and Heap-FD with non-robust design (\textit{``Heap-FD + ZF-NRD''}). To show the effectiveness of the proposed ZF-based transmission, we additionally examine the following transmission strategies:
\begin{enumerate}
	\item ``Perf. CSI + ZF:'' The perfect CSI is used to compute the ZF precoder and detector for DL and UL transmissions, respectively.
	\item The proposed heap-FD is employed for training, and then, maximum ratio transmission/combining (MRT/MRC) is applied to DL/UL transmission, called ``Heap-FD + MRT/MRC-RD.'' 
\end{enumerate}

\begin{figure}[t]
	\centering
	\includegraphics[width=0.72\columnwidth,trim={0cm 0cm 0cm 0cm}]{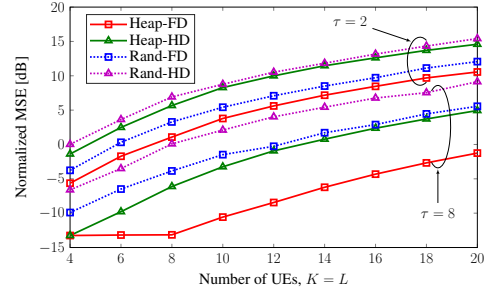}
	\vspace{-10pt}
	\caption{Normalized MSE versus the number of UEs with  $ \tau\in\{2, 8 \}$.}
	\label{fig: NMSE vs No UEs}
\end{figure}

\begin{figure}[t]
	\vspace{-10pt}
	\centering
	\includegraphics[width=0.72\columnwidth,trim={0cm 0cm 0cm 0cm}]{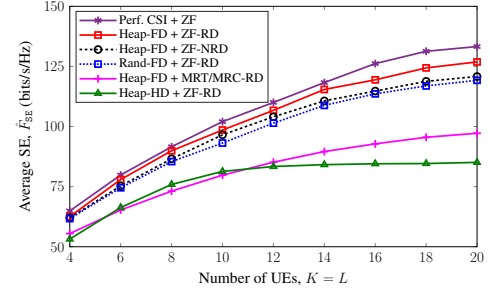}
	\vspace{-10pt}
	\caption{Average SE versus the number of UEs in CF-mMIMO with $ \tau=8 $.}
	\label{fig: SE vs No UEs}
\end{figure}

It can be easily foreseen that the quality of channel estimates mainly depends on the relationship between the number of UEs and  dimension of pilot set (or pilot length, $ \tau $). To evaluate the performance of the proposed FD training strategy, we first investigate the normalized MSE (NMSE) as a function of the number of UEs. As depicted in Fig. \ref{fig: NMSE vs No UEs}, we consider four strategies: two heap structures for pilot assignment (Heap-FD and Heap-HD) and two random pilot assignments (Rand-FD and Rand-HD). As expected, the proposed heap training schemes outperform the random ones. It can also be observed that  FD training strategies offer better performance in terms of NMSE compared to HD ones, by exploiting larger dimension of pilot sequences more efficiently. In particular, when $ K=L>\tau $, NMSE of the proposed Heap-FD  is around 5  dB and 7 dB less than Heap-HD, corresponding to $ \tau=2 $ and $ \tau= 8 $, respectively. 

It should be noted that the FD training strategy requires double training time over its HD counterpart, leading to the difference of the effective time for data transmission. The SE under imperfect CSI can be expressed as
$
	\hat{F}_{\mathtt{SE}} = \frac{\tau_{\mathtt{c}}-\tau_{\mathtt{t}}}{\tau_{\mathtt{c}}}F_{\mathtt{SE}}\bigl(\mathbf{w}, \mathbf{p},\boldsymbol{\alpha}\bigr),
$
where $\tau_{\mathtt{c}}$ and $\tau_{\mathtt{t}}$ are the coherent time and training time, respectively. We now plot the SE performance for the worst-case robust design by taking into account the CSI errors. In Fig. \ref{fig: SE vs No UEs}, we set $ \tau_{\mathtt{c}}=200 $,  $ \tau_{\mathtt{t}}=2\tau $ for  FD and $ \tau_{\mathtt{t}}=\tau $ for HD. Unsurprisingly,  Heap-FD schemes outperform the HD one, and their performance gaps are even more remarkable when the number of UEs increases, i.e., at $ K=L=20 $ gaining around 7 and 9 bits/s/Hz as compared to non-robust and Rand-FD schemes, respectively. This again demonstrates the effectiveness of the proposed Heap-based pilot assignment algorithm for FD CF-mMIMO by reaping both the advantages of  higher dimension of pilot sequences for training and FD for data transmission.

\begin{figure}[t]
	\centering
	\includegraphics[width=0.8\columnwidth,trim={0cm 0cm 0cm 0cm}]{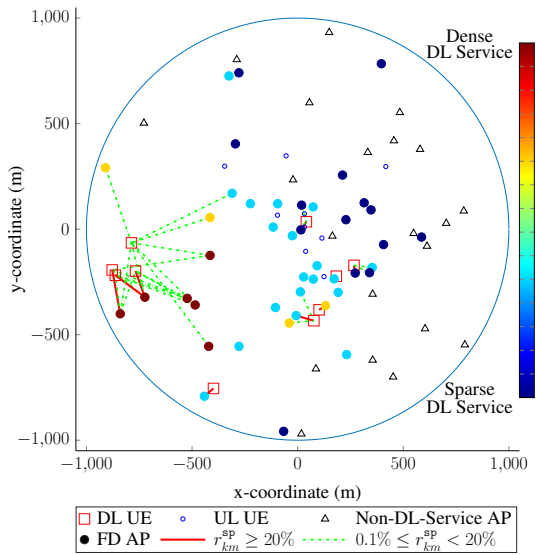}
	\vspace{-5pt}
	\caption{Service map for DL transmission, where the strong AP-DL UE connections are distinguished by $r_{km}^{\mathtt{sp}}$ (known as $ r_{\mathtt{sp}}\bigl(\mathbf{w}_{km}^{*},\mathbf{\hat{h}}_{km}^{\dl}|\mathbf{w}_{k}^{*},\mathbf{\hat{h}}_{k}^{\dl}\bigr) $ in \eqref{eq: Bfunction} for short), with $ K=L=10 $ and $\varpi=10^{-3}/M\approx0.002\%$.}
	\label{fig: DL UE Service Map}
\end{figure}

Differently from UL transmission where the APs passively receive the signals from UL UEs, in DL transmission, the CPU computes the AP-DL UE associations to decide which APs would serve DL UEs. {\hilidraf To deeply obtain insights into the proposed algorithm, Fig. \ref{fig: DL UE Service Map} presents the service map for DL UEs. In this figure, the DL service density of an AP (on the color bar) is determined by the number of DL UEs served by that AP. Fig. \ref{fig: DL UE Service Map} also indicates the significant connections between APs and DL UEs as given by $ r_{\mathtt{sp}}\bigl(\mathbf{w}_{km}^{*},\mathbf{\hat{h}}_{km}^{\dl}|\mathbf{w}_{k}^{*},\mathbf{\hat{h}}_{k}^{\dl}\bigr) $ in \eqref{eq: Bfunction}. It can be seen that the strong connections are dynamically established among APs and DL UEs in close distances (or with good channel conditions). This phenomenon further verifies the selection of favorable channels discussed in Section \ref{OptimizationProblemDesign}-B}.

\section{Conclusion}\label{Conclusion}
{\hilidraf We have studied an FD CF-mMIMO network, where  power control and AP-UE association are jointly optimized under channel uncertainty. First, we have proposed a novel and low-complexity pilot assignment algorithm based on the  heap structure to improve the quality of channel estimates. Then, we have introduced the generalized robust design taking into account the CSI errors. The special relationship between binary and continuous variables has been  exploited to devise the optimal solution for the ZF-based robust problem. Numerical results have demonstrated that our proposed algorithms outperform the existing robust designs. The effectiveness in handling the AP-DL UE associations has been also verified by a service map.}


\begingroup
\setstretch{0.99}
\bibliographystyle{IEEEtran}
\bibliography{IEEEfull}
\endgroup

\end{document}